\documentclass[11pt,letterpaper]{article}
\usepackage{fullpage}
\usepackage{amsmath, amsthm, amsfonts}
\usepackage{graphicx}
\usepackage{caption} \usepackage{subcaption, enumitem}
\usepackage{multirow}
\usepackage{geometry}
\newcommand\Tstrut{\rule{0pt}{2.9ex}}         
\newcommand\TstrutBig{\rule{0pt}{5.4ex}} 

\usepackage{hhline}

\usepackage[hyphens]{url}

\newenvironment{mymathbox}
{\par\smallskip\centering\begin{lrbox}{0}%
\begin{minipage}[c]{0.95\textwidth}}
{\end{minipage}\end{lrbox}%
\framebox[\textwidth]{\usebox{0}}%
\par\medskip
\ignorespacesafterend}

\usepackage{color}
\newcommand{\specialcell}[2][c]{%
  \begin{tabular}[#1]{@{}c@{}}#2\end{tabular}}
\newcommand{\specialcellleft}[2][c]{%
\begin{tabular}[#1]{@{}l@{}}#2\end{tabular}
}

\newtheorem*{rep@theorem}{\rep@title}
\newcommand{\newreptheorem}[2]{%
\newenvironment{rep#1}[1]{%
 \def\rep@title{#2 \ref{##1}}%
 \begin{rep@theorem}}%
 {\end{rep@theorem}}}
\makeatother
\newtheorem{theorem}{Theorem}[section]
\newreptheorem{theorem}{Theorem}
\newtheorem{definition}[theorem]{Definition}

\newreptheorem{lemma}{Lemma}
\newtheorem{problem}[theorem]{Problem}

\title{Determining Tournament Payout Structures\\ for Daily Fantasy Sports}

\author{
Christopher Musco \thanks{Massachusetts Institute of Technology. Worked completed while at Yahoo Labs. Email:  \texttt{cpmusco@mit.edu}}
\and
Maxim Sviridenko \thanks{Yahoo Labs. Email:  \texttt{sviri@yahoo-inc.com}}
\and
Justin Thaler \thanks{Georgetown University. Worked completed while at Yahoo Labs. Email: \texttt{jthaler@fas.harvard.edu}}
}
\begin{document}

\maketitle

\begin{abstract}
With an exploding global market and the recent introduction of online cash prize tournaments, fantasy sports contests are quickly becoming a central part of the social gaming and sports industries. For sports fans and online media companies, fantasy sports contests are an opportunity for large financial gains. However, they present a host of technical challenges that arise from the complexities involved in running a web-scale, prize driven fantasy sports platform.

We initiate the study of these challenges by examining one concrete problem in particular: how to algorithmically generate contest payout structures that are 1) economically motivating and appealing to contestants and 2) reasonably structured and succinctly representable. We formalize this problem and present a general two-staged approach for producing satisfying payout structures given constraints on contest size, entry fee, prize bucketing, etc.

We then propose and evaluate several potential algorithms for solving the payout problem efficiently, including methods based on dynamic programming, integer programming, and heuristic techniques. Experimental results show that a carefully designed heuristic scales very well, even to contests with over 100,000 prize winners. 

Our approach extends beyond fantasy sports -- it is suitable for generating engaging payout structures for any contest with a large number of entrants and a large number of prize winners, including other massive online games, poker tournaments, and real-life sports tournaments.
\end{abstract}
\thispagestyle{empty}
\clearpage
\setcounter{page}{1}

\section{Introduction}
\vspace{-.75em}
In many competitions, a large number of entrants compete against each other and are then {ordered} based on performance. Prize money is distributed to the entrants based on their rank in the order, with higher ranks receiving more money than lower ranks. 
The question that we are interested in is: how should prize money be distributed among the entrants?  That is, how much money should 
 go to the winner of the contest? How much to 2nd place? How much to 1,128th place? 

\vspace{-1em}
\subsection{Motivation} 
\vspace{-.65em}
We became interested in this problem in the context of daily \emph{fantasy sports}\footnote{In fantasy sports, participants build a team of real-world athletes, and the team earns points based on the actual real-world performance of the athletes. 
Traditional fantasy sports competitions run for an entire professional sports season; daily sports competitions typically run for just a single day or week.
\vspace{-.25em}}, a growing sector of online fantasy sports competitions  where users pay a fee to enter and can win real prize money. 

Daily fantasy sports were legalized in the United States in 2006 by the Unlawful Internet Gambling Enforcement Act, which classified them as \emph{games of skill}. Since then, the industry has been dominated by two companies,
FanDuel and DraftKings. In 2015, these companies collected a combined \$3 billion in entry fees, which generated an estimated \$280 million in revenue \cite{updatedrevenues}. Analysts project continued industry growth as daily fantasy attracts a growing portion of the 57 million people who actively play traditional fantasy sports  \cite{eilersResearch,fstaDemo}.

Yahoo launched a daily fantasy sports product in July 2015. Work on the contest management portion of this product has led to interesting economic and algorithmic challenges involving player pricing, revenue maximization, fill-rate prediction and of course, payout structure generation. 

\smallskip
\noindent  \textbf{The Importance of Payout Structure.}
Payout structure has been identified as an important factor in determining how appealing a competition is to users. 
Payouts are regularly discussed on forums and websites devoted to the daily fantasy sports industry, and these structures have a substantial effect on the strategies that contest entrants pursue (see, e.g., \cite{roto1,roto2,roto3,roto4,roto5}). 

Furthermore, considerable attention has been devoted to payout structures in related contexts. For example, popular articles discuss the payout structures used in World Series of Poker (WSOP) events (see Section \ref{sec:poker} for details), and at least one prominent poker tournament director has said that ``payout structure could determine whether or not a player comes back to the game.'' \cite{poker2}. 

\vspace{-1em}
\subsection{Payouts in Daily Fantasy Sports}
\vspace{-.6em}
For some types of fantasy sports contests, the appropriate payout structure is obvious. For example, in a ``Double Up'' contest, roughly half of the entrants win back twice the entry fee, while the other half wins nothing. 
However, some of the most popular contests are analogous to real-world golf and poker tournaments, in which the winner should win a very large amount of money, second place should win slightly less, and so on. We refer to such competitions as \emph{tournaments}.  

Manually determining a reasonable payout structure for each tournament offered is unprincipled and laborious, especially when prize pools and contestant counts vary widely across contests. Furthermore, given typical constraints, manually constructing even a single payout structure is difficult, even ``virtually impossible'' in the words of a WSOP Tournament Director \cite{poker2}. 
The challenge is amplified for online tournaments where the number of contestants and prizes awarded can be orders of magnitude larger than in traditional gaming and sporting: FanDuel and DraftKings run contests with hundreds of thousands of entrants and up to \$15 million in prizes.
Accordingly, our goal is to develop efficient algorithms for automatically determining payouts. 

\vspace{-1em}
\subsection{Summary of Contributions}
\vspace{-.6em}
Our contributions are two-fold. First, we (partially) formalize the properties that a payout structure for a daily fantasy tournament should satisfy.
Second, we present several algorithms for calculating such payout structures based on a general two stage framework. In particular, we present an efficient heuristic that scales to extremely large tournaments and is currently in production at Yahoo.

All methods are applicable beyond fantasy sports to any large tournament, including those for golf, fishing, poker, and online gaming (where very large prize pools are often crowd-funded \cite{crowdfunding}).

\section{Payout Structure Requirements}
\vspace{-.5em}

\label{sec:requirements}
We begin by formalizing properties that we want in a payout structure.
A few are self-evident.
\vspace{-.3em}
\begin{itemize}[itemsep=-.1em]
\item \textbf{(Prize Pool Requirement)} The total amount of money paid to the users must be equal to the Total Prize Pool. This is a hard requirement, for legal reasons: if a contest administrator says it will pay out \$1 million, it must pay out exactly \$1 million. 
\item\textbf{(Monotonicity Requirement)} The prizes should satisfy \emph{monotonicity}. First place should win at least as much as second place, who should win at least as much as third, and so on.
\end{itemize}
\vspace{-.3em}
\noindent There are less obvious requirements as well.
\vspace{-.3em}
\begin{itemize}[itemsep=-.1em]
\item \textbf{(Bucketing Requirement)} To concisely publish payout structures,  prizes should fall into a manageable set of ``buckets'' such that all users within one bucket get the same prize. It is not desirable to pay out thousands of distinct prize amounts.

\item \textbf{(Nice Number Requirement)} Prizes should be paid in amounts that are aesthetically pleasing. Paying a prize of $\$1,000$ is preferable to $\$1,012.11$, or even to $\$1,012$.  

\item \textbf{(Minimum Payout Requirement)}: It is typically unsatisfying to win an amount smaller than the entry fee paid to enter the contest. So any place awarded a non-zero amount should receive at least some minimum amount $E$. Typically, we set $E$ to be 1.5 times the entry fee.
\end{itemize}
\vspace{-.3em}
\noindent Finally, the following characteristic is desirable:
\vspace{-.3em}
\begin{itemize}[itemsep=-.1em]
\item \textbf{(Monotonic  Bucket Sizes)}: Buckets should increase in size when we move from higher ranks to lower ranks.
For example, it is undesirable for 348 users to receive a payout of \$20, 2 users to \$15, and 642 users to receive \$10. 
\end{itemize}

Of course, it is not enough to simply find a payout structure satisfying all of the requirements above. For example,
a winner-take-all payout structure satisfies all requirements, but is not appealing to entrants. 
Thus, our algorithms proceed in two stages. We first determine an ``initial'', or ideal, payout structure that
captures some intuitive notions of attractiveness and fairness. We then modify the initial payout structure \emph{as little as possible} to satisfy the requirements. 

Before discussing this process, three final remarks regarding requirements are in order.

\medskip
\noindent \textbf{Small Contests.} In tournaments with few entrants, it is acceptable to pay each winner a distinct prize. In this case, the bucketing requirement is superfluous and choosing payouts is much easier. 

\medskip
\noindent \textbf{Handling Ties.} In daily fantasy sports and other domains such as poker and golf, entrants who tie for a rank typically split all prizes due to those entrants equally. Accordingly, even if the initial payout structure satisfies the Nice Number Requirement, the actual payouts may not. However,
this is inconsequential: the purpose of the Nice Number Requirement is to ensure that aesthetically pleasing payouts are published, not to ensure that aesthetically pleasing payouts are received. 

\medskip \noindent \textbf{What are ``Nice Numbers''?}
There are many ways to define nice numbers, i.e., the numbers that we deem aesthetically pleasing. 
Our algorithms will work with any such definition, as long as it comes with an algorithm that, given a number $a \in \mathbb{R}$,
can efficiently return the largest nice number less than or equal to $a$. 
Here we give one possible definition.

\begin{definition}[Nice Number]
\label{def:nice_numbers}
A nonnegative integer $X\in Z_+$ is a ``nice number'' if $X=A\cdot 10^K$ where $K, A \in Z_+$, and $A\le 1000$ satisfies all of the following properties:
\vspace{-.3em}
\begin{enumerate}[itemsep=-.3em]
\item  if $A\ge 10$ then $A$ is a multiple of $5$;
\item  if $A\ge 100$ then $A$ is a multiple of $25$;
\item if $A\ge 250$ then $A$ is a multiple of $50$.
\end{enumerate}
\end{definition}
\vspace{-.3em}
\noindent Under Definition \ref{def:nice_numbers}, the nice numbers less than or equal to 1000 are:
\vspace{-.4em}
\begin{align*}
\{1, 2, 3, \dots, 10, 15, 20, \dots, 95,\allowbreak 100,\allowbreak 125, 150,
\dots, 225, 250, 300, 350, \dots, 950, 1000\}
\end{align*}
\vspace{-1.75em}

\noindent The nice numbers between 1000 and 3000 are
$\{1000, 1250, 1500, 1750, 2000, 2250, 2500, 3000\}$.

\subsection{Prior Work} 
\vspace{-.5em}
\label{sec:prior_work}
While our exact requirements have not been previously formalized, observation of published payout structures for a variety of contests suggests that similar guiding principals are standard. Moreover, manually determining payouts to match these requirements is, anecdotally, not an easy task. For example, it is not hard to find partially successful attempts at satisfying nice number constraints.
\vspace{-1em}
\begin{figure}[h]
\centering
\includegraphics[width=.55\textwidth]{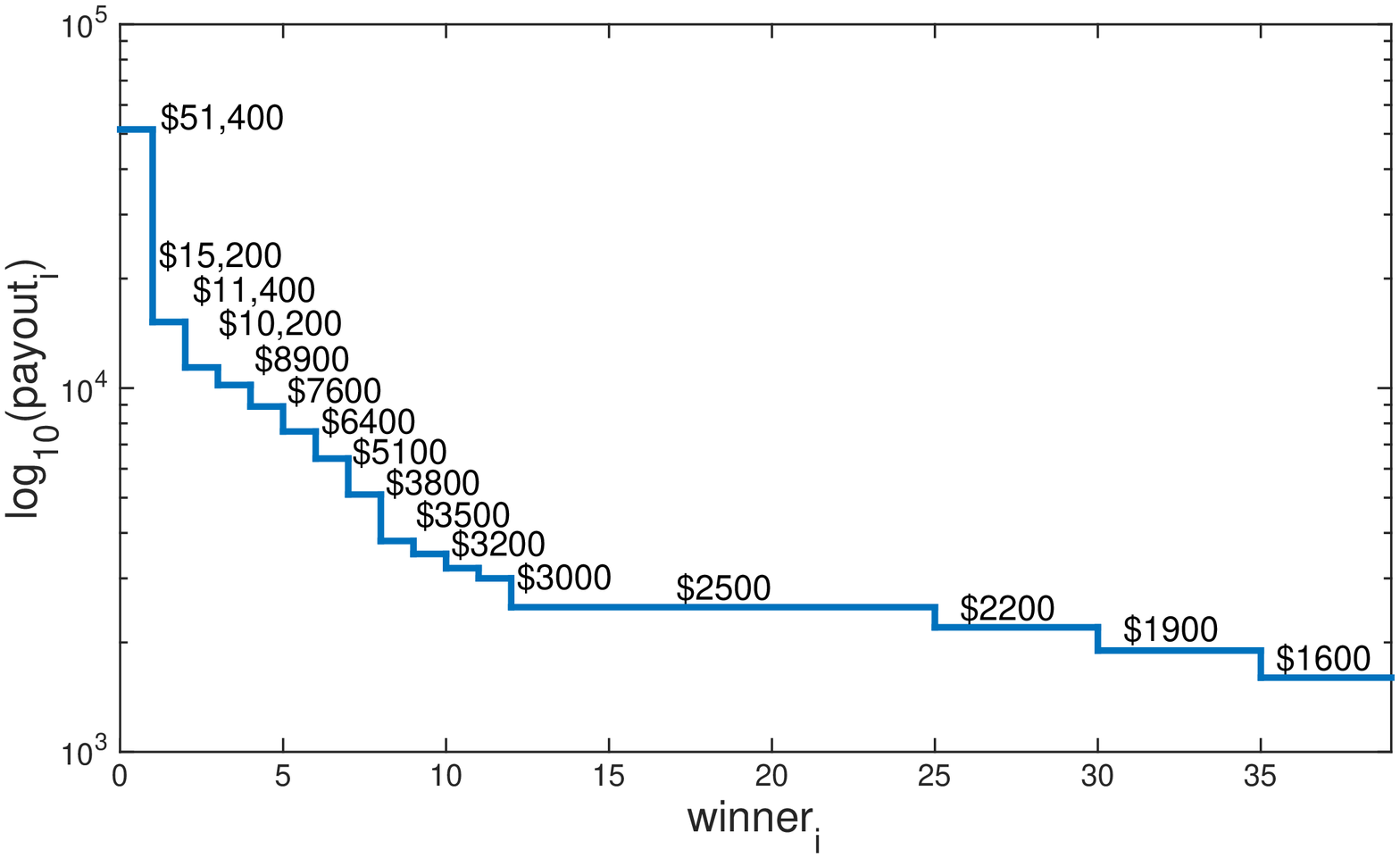}
\vspace{-1em}
\caption{Bassmaster Open payouts are not quite ``nice numbers'' \protect\cite{bassmasterInfo}.}
\label{fig:bmaster_failure}
\vspace{-.75em}
\end{figure}

For example, consider the 2015 Bassmaster Open fishing tournament, which paid out a top prize of \$51,400, instead of the rounder \$51,000 or \$50,000 (see Figure \ref{fig:bmaster_failure}). 
The Bassmaster payout structure also violates our constraint on bucket size monotonicity. 
Similar ``partially optimized'' payout structures can be found for poker \cite{pokerPayoutsFailure,poker2}, golf  \cite{usopenprizes}, and other tournaments.

\smallskip \noindent
\textbf{Payout Structures for Poker.} \label{sec:poker}
Several popular articles describe the efforts of the World Series Of Poker (WSOP),  in conjunction with Adam Schwartz of Washington and Lee University,
to develop an algorithm for determining payout structures for their annual ``Main Event'' \cite{poker1,poker2}.
Schwartz's solution, which was based on Newton's Method, was combined with manual intervention to determine the final payout structure in 2005.
WSOP still utilized considerable manual intervention when determining payouts until an update of the algorithm in 2009 made doing so unnecessary.

While a formal description of the algorithm in unavailable, it appears to be very different from ours; Schwartz has stated that their solution attempts to obtain payout structures with a ``constant second derivative'' (our solutions do not satisfy this property). Their work also departs qualitatively from ours in that they do not consider explicit nice number or bucket size requirements.

\smallskip \noindent
\textbf{Piecewise Function Approximation.}
As mentioned, our payout structure algorithms proceed in two stages. An initial payout curve is generated with a separate payout for every winning position. The curve is then modified to fit our constraints, which requires bucketing payouts so that a limited number of distinct prizes are paid. We seek the bucketed curve closest to our initial payout curve.

This curve modification task is similar to the well studied problem of optimal approximation of a function by a histogram (i.e., a piecewise constant function). This problem has received considerable attention, especially for applications to database query optimization \cite{Ioannidis:2003}. While a number of algorithmic results give exact and approximate solutions \cite{Jagadish:1998,Guha:2006}, unfortunately no known solutions easily adapt to handle our additional constraints beyond bucketing.


\vspace{-1em}
\section{Our Solution} 
\vspace{-.5em}
Let $B$ denote the total prize pool, $N$ denote the number of entrants who should win a non-zero amount, and  $P_i$ denote the prize that we decide to award to place $i$. In general, $B$, $P_1$ and $N$ are user-defined parameters. 
$P_1$ can vary widely in fantasy sports contests, anywhere from $.05\cdot B$ to nearly $.5\cdot B$, but $.15 \cdot B$ is a standard choice (i.e., first place wins 15\% of the prize pool).  Typically, $N$ is roughly 25\% of the total number of contest entrants, although it varies as well.

Our solution proceeds in two stages. We first determine an initial payout structure that does not necessarily satisfy  Bucketing or Nice Number requirements. The initial payout for place $i$ is denoted $\pi_i$. We then modify the payouts to satisfy our requirements. In particular, we search for feasible payouts $P_1, \ldots, P_N$ that are nearest to our initial payouts in terms of sum-of-squared error.

\vspace{-1em}
\subsection{Determining the Initial Payout Structure}
\label{sec:initial_structure} 
\vspace{-.5em}
First, to satisfy the Minimum Payout Requirement, we start by giving each winning place $E$ dollars. This leaves $B-N \cdot E$ additional dollars to disperse among the $N$ winners. How should we do this?

We have decided that it is best to disperse the remaining budget according to a \emph{power law}. That is, the amount of the budget that is given to place $i$ should be proportional to $1/i^{\alpha}$ for some fixed constant $\alpha > 0$. It is easy to see that for \emph{any} positive value of $\alpha$, the resulting payouts will satisfy the Monotonicity Requirement, but there is a unique value of $\alpha$ ensuring that the payouts sum to exactly the Total Prize Pool $B$. 
Specifically, we need to choose the exponent $\alpha$ to satisfy 
\vspace{-.5em}
$$B-N\cdot E=\sum_{i=1}^N\frac{P_1-E}{i^\alpha}.$$ 
\vspace{-1.25em}

\noindent We can efficiently solve this equation for $\alpha$ to additive error less than $.01$ via binary search. We then define the ideal payout to place $i$ 
to be $\pi_i := E +  \frac{P_1-E}{i^\alpha}$. This definition ensures both that first place gets paid exactly $\$P_1$, and that the sum of all of the tentative payouts is exactly $B$.

\vspace{-1em}
\begin{figure}[h]
\centering
\begin{subfigure}{0.48\textwidth}
\centering
\includegraphics[width=.75\textwidth]{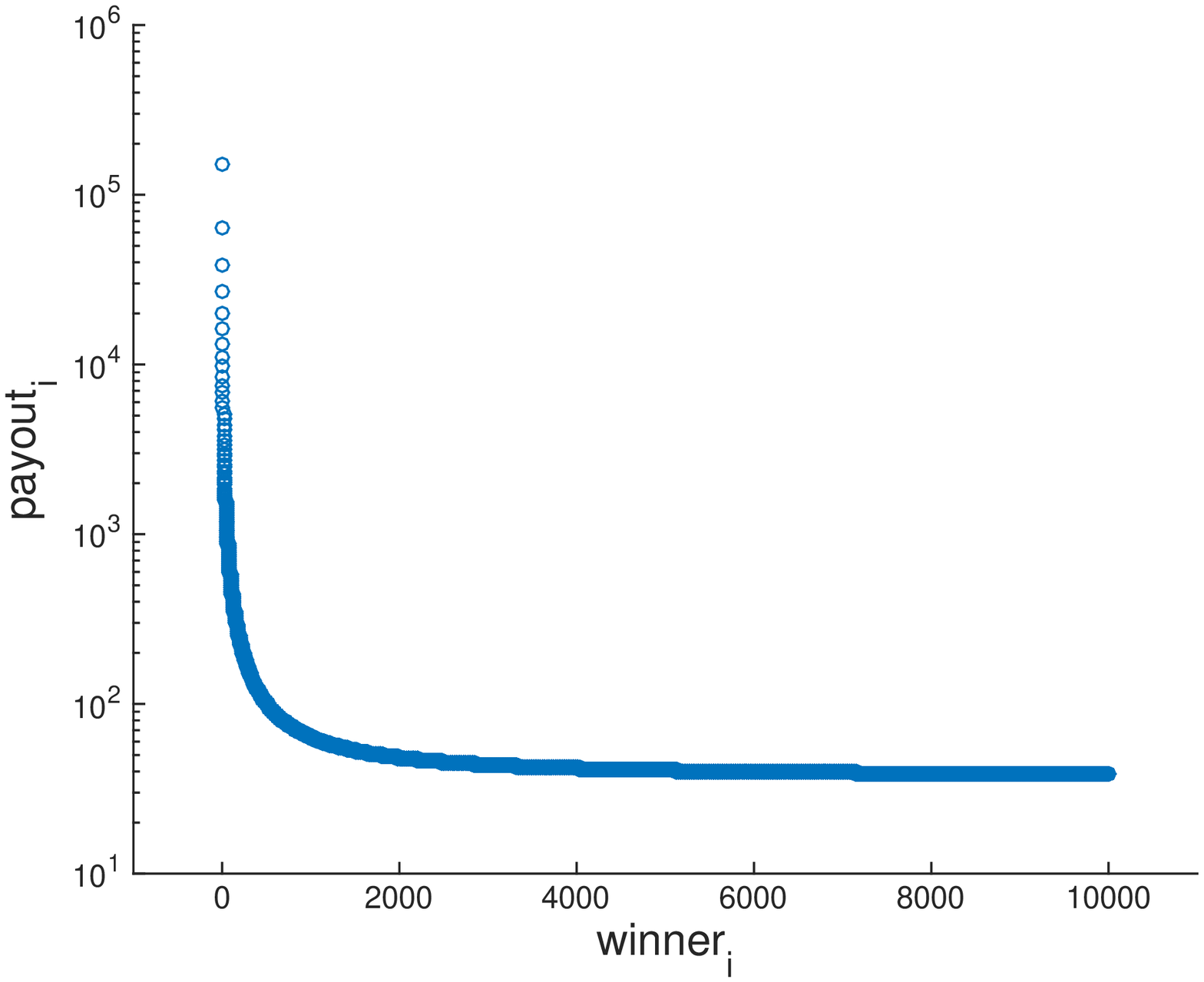} 
\vspace{-.5em}
\caption{Ideal payouts using power law method.}
\label{fig1}
\end{subfigure}%
\hfill
\begin{subfigure}{0.48\textwidth}
\centering
\includegraphics[width=.75\textwidth]{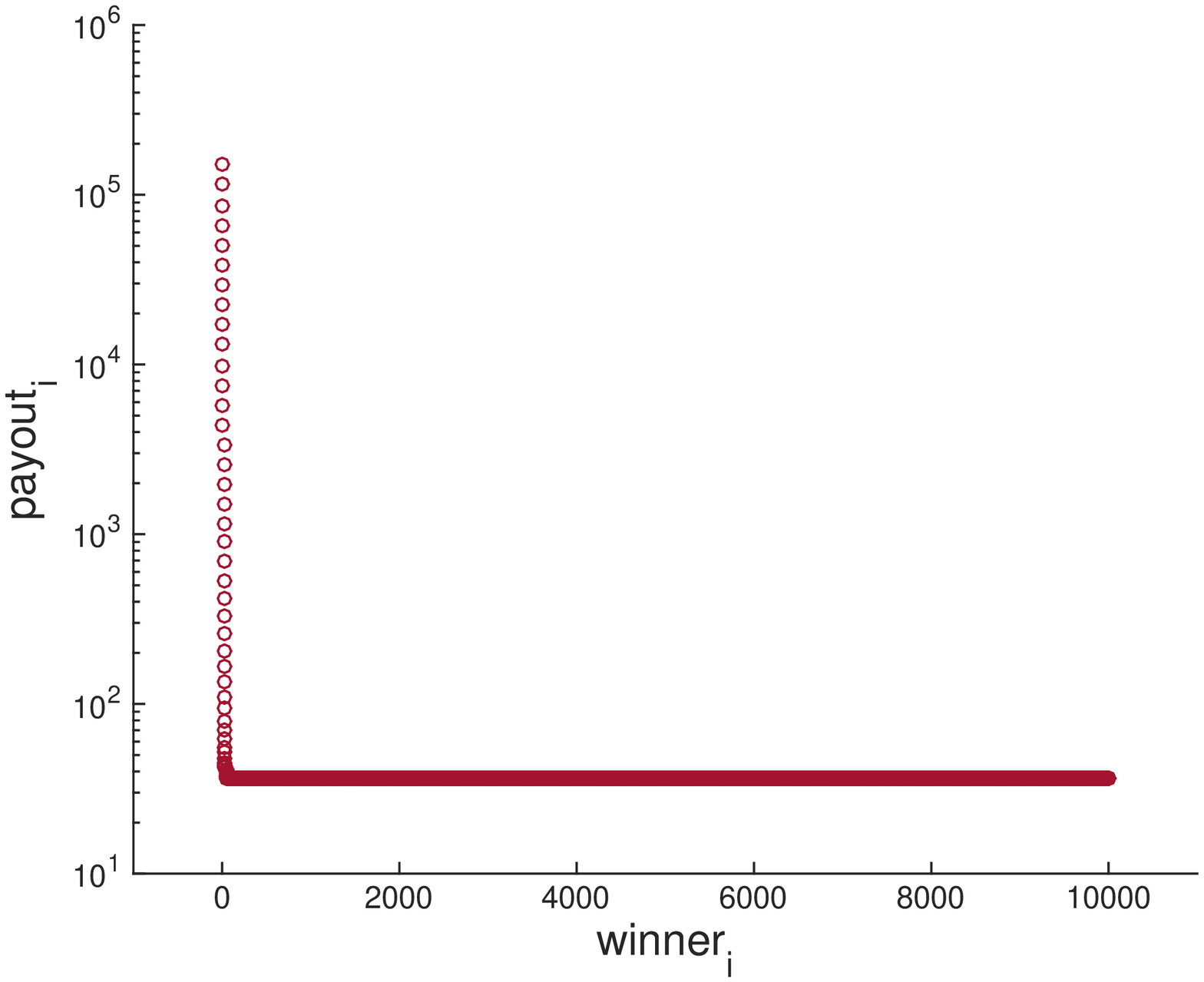}
\vspace{-.5em}
\caption{Ideal payouts using exponential distribution.}
\label{fig2}
\end{subfigure}
\vspace{-.25em}
\caption{Possible initial payout structure when $N=\text{10,000}$, $P_1=\text{150,000}$, and $B=\text{1 million}$.}
\label{examplefig}
\vspace{-1em}
\end{figure}

Why use a power law? Empirically, a power law ensures that as we move from 1st place to 2nd to 3rd and so on, payouts drop off at a nice pace: fast enough that top finishers are richly rewarded relative to the rest of the pack, but slow enough that users in, say, the 10th percentile still win a lot of money. A power law curve also encourages increased prize differences between higher places, a property cited as a desirable by WSOP organizers \cite{poker2}.


For illustration, consider a tournament where 40,000 entrants vie for \$1 million. 
If 1st place wins 15\% of the prize pool and 25\% of entrants should win a non-zero prize, then $P_1 =\$150,000$ and $N=10,000$. 
Figure \ref{fig1} reveals the initial payouts determined by our power law method. 

Figure \ref{fig2} reveals what initial payouts would be if we used an exponential distribution instead, with prizes proportional to $1/\alpha^i$ rather than $1/i^\alpha$. Such distributions are a popular choice for smaller tournaments, but the plot reveals that they yield much more top-heavy payouts than a power law approach. 
In our example, only the top few dozen places receive more than the minimum prize.

\vspace{-.75em}
\begin{figure}[h]
\centering
\begin{subfigure}{0.48\textwidth}
\centering
\includegraphics[width=.75\textwidth]{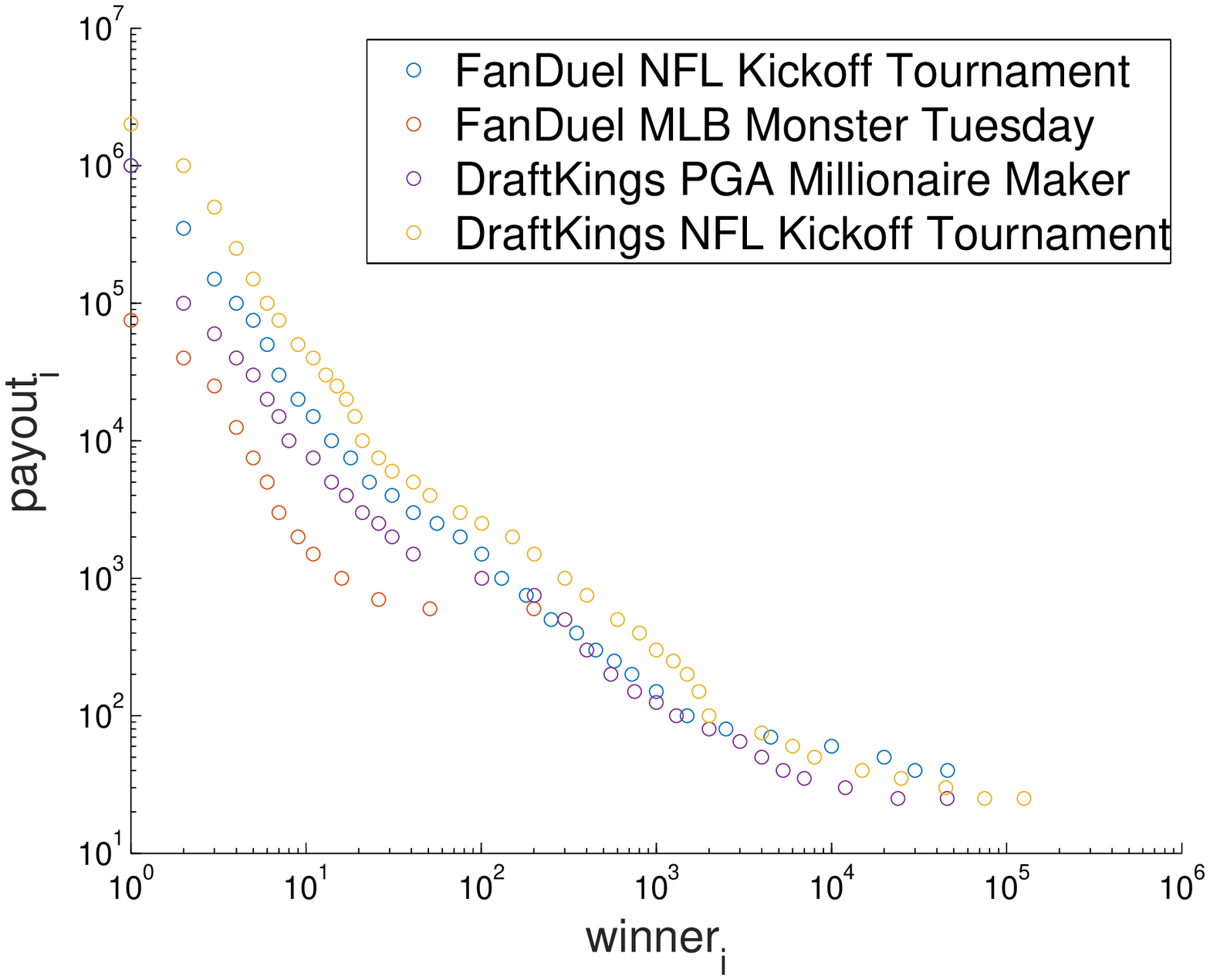} 
\vspace{-.5em}
\caption{Large fantasy sports tournaments.}
\label{empirical_payouts:sub1}
\end{subfigure}%
\hfill
\begin{subfigure}{0.48\textwidth}
\centering
\includegraphics[width=.75\textwidth]{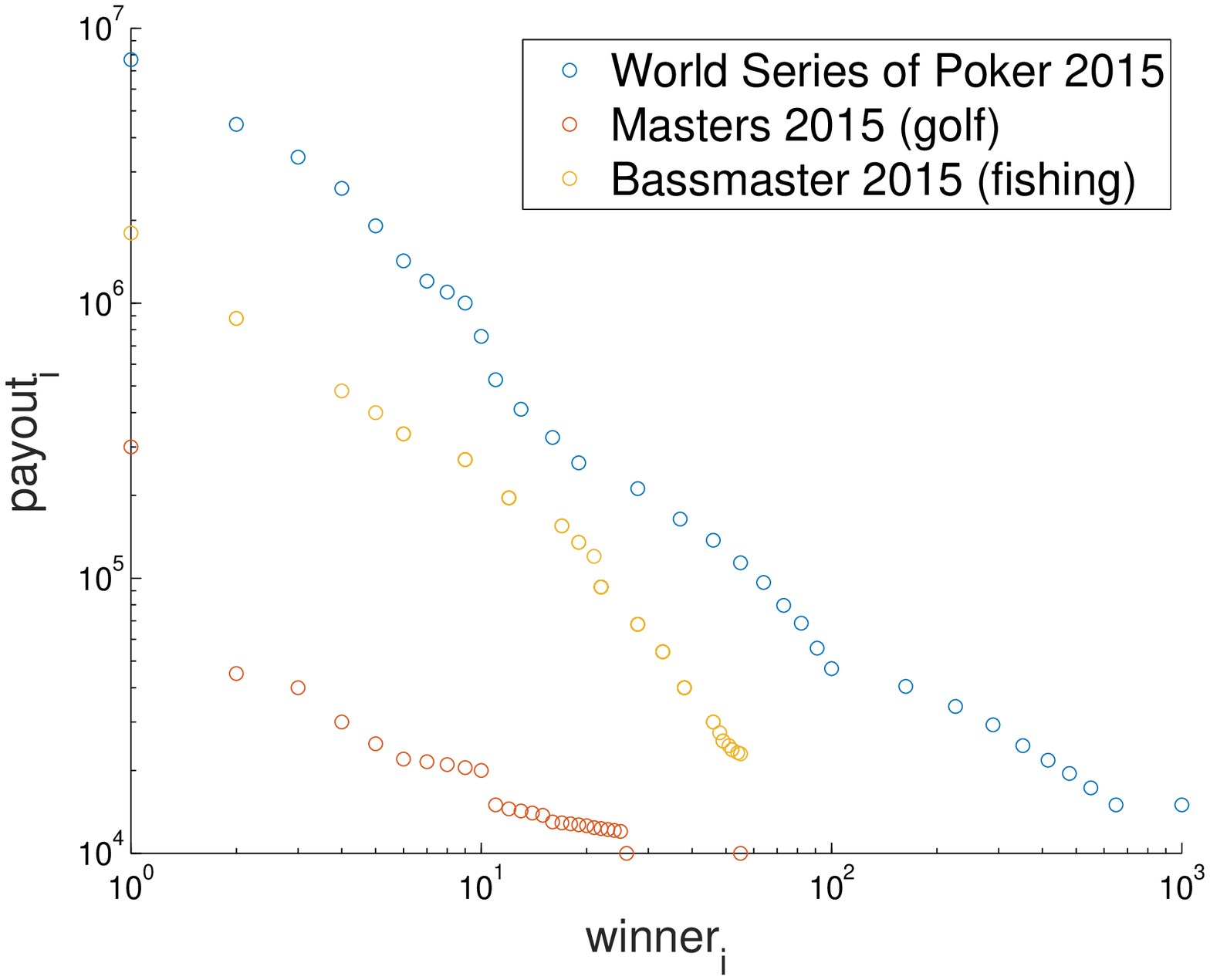}
\vspace{-.5em}
\caption{Other well known tournaments.}
\label{empirical_payouts:sub2}
\end{subfigure}
\vspace{-.7em}
\caption{Log plots of existing tournament payout structures indicate linear structure.}
\label{fig:empirical_payouts}
\vspace{-.75em}
\end{figure}

As further justification, we check in Figure \ref{fig:empirical_payouts} that a power law roughly models payouts for existing fantasy sports contests and other large tournaments. Since $\log(i^\alpha) = \alpha\log(i)$, plotting payouts with both axes on a logarithmic scale will reveal a linear trend when the payout distribution approximates a power law. This trend is visible in all of the tournaments checked in Figure  \ref{fig:empirical_payouts}.

Obtaining quantitative goodness-of-fit metrics for the power law can be difficult \cite{clauset2009power} especially given the heavy ``value binning'' present in payout structures \cite{virkar2014power}. Nevertheless, standard fitting routines \cite{powerLawCode} yield an average $p$-value of $.237$ for the Figure \ref{fig:empirical_payouts} tournaments. This value is in line with calculated $p$-values for well known empirical power law distributions with binned values \cite{virkar2014power}.

Overall, the power law's simplicity, empirical effectiveness, and historical justification make it an ideal candidate for generating initial prize values.

\vspace{-1em}
\subsection{Satisfying the Remaining Requirements}
\label{sub:rounding}
\vspace{-.5em}
Now that we have an initial payout structure $(\pi_1, \dots, \pi_N)$, our goal becomes to identify a related payout structure that is ``close'' to this ideal one, but satisfies our Bucketing and Nice Number requirements. To measure ``closeness'' of payout structures we have found that (squared) Euclidean distance works well in practice: we define the distance between two payout structures $(P_1, \dots, P_N)$ and $(Q_1, \dots, Q_N)$ to be $\sum_{i=1}^N (P_i - Q_i)^2$. 
This metric is equivalent to the popular ``V-optimal'' measure for approximation of a function by a general histogram \cite{Ioannidis:1995}.

With a metric in hand, our task is formalized as an optimization problem, Problem  \ref{main_problem}.
Specifically, suppose we are given a target number of buckets $r$. Then the goal is to partition the set  $\{1,\ldots, N\}$ into $r$ buckets $S_1,\ldots,S_r$ (each containing a set of consecutive places), and to choose a set of payouts $\Pi_1, \dots, \Pi_r$ for each bucket, forming a final solution $(S_1, \dots, S_r, \Pi_1, \dots, \Pi_r)$.  

\vspace{-.5em}
\begin{figure*}[h]
\begin{mymathbox}
\begin{problem}[Payout Structure Optimization]
\label{main_problem}
For a given set of ideal payouts $\{\pi_1, \ldots, \pi_N\}$ for $N$ contest winners, a total prize pool of $B$, and minimum payout $E$, find $(S_1, \dots, S_r, \Pi_1, \dots, \Pi_r)$ to optimize:
\vspace{-.75em}
\begin{eqnarray*}
\min \sum_{j=1}^r \sum_{i\in S_j}(\pi_i-\Pi_{j})^2&\text{subject to:}&\\
E \le \Pi_r<\Pi_{r-1}< \dots< \Pi_1,& & \hspace{-0mm} \text{(Monotonicity \& Min. Payout Requirements)}\\
\sum_{j=1}^r\Pi_j|S_j|=B,&& \text{(Prize Pool Requirement)}\\
\Pi_j \mbox{ is a ``nice number'' },& j \in [r]\footnote{For ease of notation we use $[T]$ to denote the set of integers $\{1,2,\ldots,T\}$.}&\text{(Nice Number Requirement)}\\
\sum_{j=1}^r |S_j|=N,&& \text{(Ensure Exactly $N$ Winners)}\\
0\leq |S_1| \leq |S_{2}| < \dots \leq |S_r|,&&\text{(Monotonic Bucket Sizes}\footnote{Setting $S_1 = \emptyset$,  $S_2 = \emptyset$, etc. chooses a payout structure with fewer buckets than the maximum allowed.}\text{)}
\end{eqnarray*}
$\textit{where } S_j = \left\{\sum_{i < j} |S_i|+1, \sum_{i < j} |S_i|+2, \ldots, \sum_{i \leq j} |S_i|\right\} \text{ for } j \in [r]$.
\end{problem}
\end{mymathbox}
\vspace{-1.25em}
\end{figure*}

One advantage of our approach is that Problem \ref{main_problem} is agnostic to the initial curve $(\pi_1, \dots, \pi_N)$. Our algorithms could just as easily be applied to an ideal payout curve generated, for example, using the ``constant second derivative'' methodology of the World Series of Poker \cite{poker2}.

\smallskip
\noindent  {\textbf{Problem Feasibility.}
Note that, for a given set of constraints, Problem \ref{main_problem} \emph{could} be infeasible as formulated: it may be that no assignment to $(S_1, \dots, S_r, \Pi_1, \dots, \Pi_r)$ satisfies the Nice Number and Bucket Size Monotonicity constraints while giving payouts that sum to $B$. 
While the problem is feasible for virtually all fantasy sports tournaments tested (see Section \ref{sec:experiments}), we note that it is easy to add more flexible constraints that always yield a feasible solution. 
For example, we can soften the requirement on $N$ by adding a fixed objective function penalty for extra or fewer winners.

With this change, Problem 3.1 is feasible whenever B is a nice-number: we can award a prize of B to first place and prizes of zero to all other players. Experimentally, there are typically many feasible solutions, the best of which are vastly better than the trivial ``winner-take-all'' solution.

\smallskip
\noindent  {\textbf{Exact Solution via Dynamic Programming.}
While Problem \ref{main_problem} is complex, when feasible it can be solved exactly in pseudo-polynomial time via multi-dimensional dynamic programming. 
The runtime of the dynamic program depends on the number of potential prize assignments, which includes all of the nice numbers between $E$ and $B$. Since many reasonable definitions (including our Definition \ref{def:nice_numbers}) choose nice numbers to spread out exponentially as they increase, 
we assume that this value is bounded by $O(\log B)$.

\vspace{-.25em}
\begin{theorem}[Dynamic Programming Solution]
\label{thm:dynamic}
Assuming that there are $O(\log B)$ nice numbers in the range $[E,B]$, then Problem \ref{main_problem} can be solved in pseudo-polynomial time $O(rN^3B\log^2B)$.
\end{theorem}
\vspace{-.25em}
A formal description of the dynamic program and short proof of its correctness are included in Appendix \ref{app:dyno_proof}.
Unfortunately, despite a reasonable theoretical runtime, the dynamic program requires $O(rN^2 B \log B)$ space, which quickly renders the solution infeasible in practice.

\vspace{-1em}
\section{Integer Linear Program}
\label{sec:offtheshelf}
\vspace{-.5em}

Alternatively, despite its complexity, we show that   it is possible to formulate Problem \ref{main_problem} as a standard integer linear program. Since integer programming is computationally hard in general, this does not immediately yield an efficient algorithm for the problem. However, it does allow for the application of off-the-shelf optimization packages to the payout structure problem. 

\vspace{-.75em}
\begin{figure*}[h]
\begin{mymathbox}
\begin{problem}[Payout Structure Integer Program]\label{integer_program}
For a given set of ideal payouts $\{\pi_1, \ldots, \pi_N\}$, a total prize pool of $B$, a given set of acceptable prize payouts $\{p_1 > p_2 > \ldots > p_m\}$, and an allowed budget of $r$ buckets solve:
\vspace{-.5em}
\begin{eqnarray*}
\min \sum_{i\in[N],j\in[r],k\in[m]} x_{i,j,k}\cdot (\pi_i - p_k)^2&\text{subject to:}&\\
\textbf{Problem constraints:}\\
\sum_{k\in[m]} (k+1/2)\cdot\tilde{x}_{j,k}  - k\cdot\tilde{x}_{j+1,k} \leq 0,&j\in[r-1], &\text{(Monotonicity Requirements)}\\
\sum_{i\in[N],j\in[r],k\in[m]} x_{i,j,k} \cdot p_k=B,&& \text{(Prize Pool Requirement)}\\
\sum_{i\in[N], k\in[m]} x_{i,j,k}  - x_{i,j+1,k} \leq 0,&j\in[r-1],&\text{(Monotonic Bucket Sizes)}\\
\textbf{Consistency constraints:}\\
\sum_{j\in[r],k\in[m]} x_{i,j,k} = 1,&i\in[N],&\text{(One Bucket Per Winner)}\\
\sum_{k\in[m]} \tilde{x}_{j,k} \leq 1,&j\in[r],&\text{(One Prize Per Bucket)}\\
\tilde{x}_{j,k}  - x_{i,j,k} \geq 0,&i\in[N], j\in[r], k\in[m], &\text{(Prize Consistency)}
\end{eqnarray*}
\end{problem}
\end{mymathbox}
\vspace{-1.25em}
\end{figure*}

To keep the formulation general, assume that we are given a fixed set of acceptable prize payouts, $\{p_1 > p_2 > \ldots > p_m\}$. These payouts may be generated, for example, using Definition \ref{def:nice_numbers} for nice numbers. In our implementation, the highest acceptable prize is set to $p_1 = P_1$, where $P_1$ is the pre-specified winning prize. Additionally, to enforce the minimum payout requirement, we chose $p_m \geq E$. Our integer program, formalized as Problem \ref{integer_program}, involves the following variables:
\vspace{-.5em}
\begin{itemize}[itemsep=-.3em]
\item $N\times r \times m$ \emph{binary} ``contestant variables'' $x_{i,j,k}$. In our final solution, $x_{i,j,k} = 1$ if and only if contestant $i$ is placed in prize bucket $S_j$ and receives payout $p_k$.
\item $r\times m$ \emph{binary} ``auxiliary variables'' $\tilde{x}_{j,k}$.
$\tilde{x}_{j,k} = 1$ if and only if bucket $S_j$ is assigned payout $p_k$. Constraints ensure that $x_{i,j,k}$ only equals $1$ when $\tilde{x}_{j,k} = 1$. If, for a particular $j$, $\tilde{x}_{j,k} = 0$ for all $k$ then  $S_j$ is not assigned a payout, meaning that the bucket is not used. 
\end{itemize}
\vspace{-.5em}

It is easy to extract a payout structure from any solution to the integer program. Showing that the payouts satisfy our Problem \ref{main_problem} constraints is a bit more involved. A proof is in Appendix \ref{app:ip_proof}. 
%

\vspace{-1em}
\section{Heuristic Algorithm}
\label{sec:heuristic}
\vspace{-.5em}
Next we describe a heuristic algorithm for Problem \ref{main_problem} that is used in production at Yahoo. The algorithm is less rigorous than our integer program and can potentially generate payout structures that violate constraints. However, it scales much better than the IP and experimentally produces stellar payouts. The heuristic  proceeds in four stages.

\smallskip
\noindent  \textbf{Stage 1: Initialize Bucket Sizes.}
First the algorithm chooses tentative bucket sizes $|S_1|\le \ldots\le |S_r|$. We set $|S_1|=|S_2|=|S_3|=|S_4|=1$. The choice to use 4 ``singleton buckets'' by default is flexible: the algorithm can easily accommodate more. If $N-\sum_{i=1}^4|S_i|=1$, we define $|S_5|=1$ and stop. Otherwise we set $|S_t|=\lceil \beta\cdot |S_{t-1}|\rceil $ where $\beta\ge 1$ is a parameter of the heuristic.  
The algorithm stops when $\lceil\beta^2 |S_t| \rceil + \lceil \beta |S_t| \rceil +\sum_{i=1}^t|S_i|>N$ and $\lceil \beta |S_t| \rceil+\sum_{i=1}^t|S_i|\le N$. We define 
$$|S_{t+1}|=\left \lfloor N-\sum_{i=1}^t|S_i|/2\right\rfloor,\ 
|S_{t+2}|=\left \lceil \frac{N-\sum_{i=1}^t |S_i|}{2}\right\rceil.$$
An initial value for $\beta$ can be determined by solving $\beta + \beta^2 + \ldots +  \beta^{r-4} = N - 4$ via binary search. If the heuristic produces more than $r$ buckets using the initial $\beta$, we increase $\beta$ and iterate. 

\smallskip
\noindent  \textbf{Stage 2: Initialize Prizes.}
Next, we begin by rounding the first tentative prize, $\pi_1$, down to the nearest nice number. The difference between $\pi_1$ and the rounded number is called the {\it leftover}, which we denote by $L$.
Iteratively, for each bucket $S_2, \ldots, S_t$ we sum all of the tentative prizes in the bucket with the leftover from the previous buckets. Let $R_t$ be this number and define $\Pi_t$ to equal $R_t/|S_t|$ rounded down to the nearest nice number. If $\Pi_t$ is greater than or equal to the prize in bucket $t-1$, $\Pi_{t-1}$, we simply merge all member of $S_t$ into $S_{t-1}$, assigning them prize $\Pi_{t-1}$.
At the end we may have some non-zero leftover $L$ remaining from the last bucket.

\smallskip
\noindent \textbf{Stage 3: Post-Process Non-monotonic Bucket Sizes.}
Although the initial buckets from Step 1 increase in size, potential bucket merging in Step 2 could lead to violations in the Monotonic Bucket Size constraint. 
So, whenever $|S_t|$ is larger than $|S_{t+1}|$, we shift users to bucket $t+1$ until the situation is corrected. As a result we decrease the prizes for some users and increase the leftover $L$. We repeat this process starting from $S_1$ and ending at our lowest prize level bucket.
 
\smallskip
\noindent  \textbf{Stage 4: Spend Leftover Funds.}
Finally, we modify payouts to spend any leftover $L$.  We first spend as much as possible on singleton buckets 2 through 4. We avoid modifying first prize because it is often given as a hard requirement -- e.g. we want pay exactly \$1 million to the winner. In order from $i=2$ to $4$ we adjust $\Pi_i$ to equal $\min\{\Pi_i+L,(\Pi_{i-1}+\Pi_i)/2\}$, rounded down to a nice number. This spends as much of $L$ as possible, while avoiding too large an increase in each prize. 

If $L$ remains non-zero, we first try to adjust only the final (largest) bucket, $S_k$.
If $L\ge |S_k|$ then we set $\Pi_k=\Pi_k+1$ and $L=L-|S_k|$, i.e. we increase the prize for every user in $S_k$ by $1$. Note that this could lead to nice number violations, which are not corrected. We repeat this process (possibly merging buckets) until $L<|S_k|$. If at this point $L$ is divisible by $\Pi_k$ we increase $|S_k|$ by $L/\Pi_k$ (thereby increasing the number of users winning nonzero prizes beyond $N$). 

If $L$ is not divisible by $\Pi_k$, we rollback our changes to the last bucket and attempt to spend $L$ on the last \emph{two} buckets. Compute the amount of money available, which is the sum of all prizes in these buckets plus $L$. Fix the last bucket prize to be the minimal possible amount, $E$. Enumerate over possible sizes and integer prize amounts  for the penultimate bucket, again ignoring nice number constraints. If the last bucket can be made to have integer size (with payout $E$), store the potential solution and evaluate a ``constraint cost'' to penalize any constraint violations. The constraint cost charges $100$ for each unit of difference if the number of winners is less than $N$, $1$ for each unit of difference if the number of winners is larger than $N$, and $10$ for each unit of violation in bucket size monotonicity. From the solutions generated return the one with minimal constraint cost.

\vspace{-1em}
\section{Experiments}
\vspace{-.5em}
\label{sec:experiments}

We conclude with experiments that confirm the effectiveness of both the integer program (Section \ref{sec:offtheshelf}) and our heuristic (Section \ref{sec:heuristic}). Both algorithms were implemented in Java and tested on a commodity laptop with a 2.6 GHz Intel Core i7 processor and 16 GB of 1600 MHz DDR3 memory.
For the integer program, we employ the off-the-shelf, open source GNU Linear Programming Kit (GLPK), accessed through the SCPSolver front-end \cite{GLPK,scpsolver}. The package uses a branch-and-cut algorithm for IPs with a simplex method for underlying linear program relaxations.

We construct experimental payout structures for a variety of daily fantasy tournaments from Yahoo, FanDuel, and DraftKings and test on non-fantasy tournaments as well. 
For non-Yahoo contests, $P_1$ is set to the published winning prize or, when necessary, to a nearby nice number. The maximum number of buckets $r$ is set to match the number of buckets used in the published payout structure.
For fantasy sports, $E$ is set to the nearest nice number above $1.5$ times the entry fee. For all other contests  (which often lack entry fees or have a complex qualification structure) $E$ is set to the nearest nice number above the published minimum prize.

\smallskip
\noindent  \textbf{Quantitative Comparison.}
Our results can be evaluated by computing the Euclidean distance between our ideal pay out curve, $\{\pi_1, \ldots, \pi_n\}$, and the bucketed curve $\{\Pi_1, \ldots \Pi_m\}$. In heuristic solutions, if $m$ does not equal $n$, we extend the curves with zeros to compute the distance (which penalizes extra or missing winners). Our experimental data is included in Table \ref{tab:main_table}. Entries of ``--'' indicate that the integer program did not run to completion, possibly because no solution to Problem \ref{integer_program} exists. The cost presented is the sum of squared distances from the bucketed payouts to ideal power law payouts, as defined in Problem \ref{main_problem}. Note that we do not provide a cost for the \emph{source} payouts since, besides the Yahoo contests, these structures were not designed to fit our proposed ideal payout curve and thus cannot be expected to achieve small objective function values.
\vspace{-1.25em}
\begin{table*}[h]
\caption{Accuracy and Runtime (in milliseconds) for Integer Program (IP) vs. Heuristic (Heur.)\label{tab:main_table}}{%
\tiny
\centering 
\vspace{-.75em}
\begin{tabular}{||c|c|c|c|c|c||c|c||c|c|c||} 
\hhline{|t:======:t:==:t:===:t|}
Source&\specialcell{Prize Pool} & \specialcell{Top Prize} & \specialcell{Min. Prize} &\specialcell{\# of\\ Winners} & \specialcell{\# of\\ Buckets}& \specialcell{IP \\ Cost} & \specialcell{IP Time \\(ms)} & \specialcell{Heur. \\Cost} & \specialcell{Heur. \\Time (ms)} & \specialcell{Heur. \\ Extra \\ Winners} \\
\hhline{||-|-|-|-|-|-||-|-||---||}
Yahoo & 90 & 25 & 2& 30 & 7& .89 &7.6k &2.35 & 1 & 0 \\ 
Yahoo & 180 & 55 & 3& 30 & 10& 2.82 &725k &3.44 & 1 & 0\\ 
DraftKings & 500 & 100 & 8& 20 & 10& 6.15 &2.1k &9.21 & 1 &0 \\ 
Yahoo & 2250 & 650 & 150& 7 & 7& 32.4 &4.0k &187.4 & 1 &0 \\
Yahoo & 3000 & 300 & 2& 850 & 25& -- &-- &86.9 & 7 & 2 \\ 
FanDuel & 4000 & 900 & 50& 40 & 12& 20.7 &3716k &58.2 & 2 &1 \\ 
FanDuel & 4000 & 800 & 75& 16 & 7& 46.6 &2.9k &230.1 & 1 &4 \\ 
DraftKings & 5000 & 1250 & 150& 11 & 8& 52.5 &6.8k &123.5 & 1 &0 \\ 
Yahoo & 10000 & 1000 & 7& 550 & 25& -- &-- &97.3 & 8 &1 \\ 
DraftKings & 10000 & 1500 & 75& 42 & 12& 61.3 &1291k &173.7 & 2 &0 \\ 
FanDuel & 18000 & 4000 & 150& 38 & 10& 161.8 &131k &347.0 & 5 &0 \\ 
FanDuel & 100000 & 10000 & 2& 23000 & 25& -- &-- &3.1k & 152 &34 \\ 
Bassmaster & 190700 & 50000 & 2000& 40 & 15& -- &-- & 3.5k\textsuperscript{*} & 3 &0 \\ 
Bassmaster & 190000\textsuperscript{$\dagger$}\ & 50000 & 2000& 40 & 15& 2.5k &3462k & 2.8k & 1 &0 \\ 
FLW Fishing & 751588 & 100000 & 9000& 60 & 25& -- &-- & 6.0k\textsuperscript{*} & 3 &0 \\ 
FLW Fishing & 751500\textsuperscript{$\dagger$}\ & 100000 & 9000& 60 & 25& -- &-- & 6.0k & 2 &0 \\ 
FanDuel & 1000000 & 100000 & 15& 16000 & 25& -- &-- &5.3k & 203 &7 \\ 
DraftKings & 1000000 & 100000 & 5& 85000 & 40& -- &-- & 25.9k & 1.2k &0 \\ 
Bassmaster & 1031500 & 30000 & 10000& 55 & 25& -- &-- & 13.5k\textsuperscript{*} & 14 &0 \\ 
FanDuel & 5000000 & 1000000 & 40& 46000 & 30& -- &-- &44.3k & 1.0k &0 \\ 
PGA Golf & 9715981 & 1800000 & 20000& 69 & 69& -- &-- & 254.5k\textsuperscript{*} & 24 &0 \\ 
PGA Golf & 1000000\textsuperscript{$\dagger$}\ & 1800000 & 20000& 75 & 75& -- &-- & 215.9k\textsuperscript{*} & 23 &9 \\ 
DraftKings & 10000000 & 2000000 & 25& 125000 & 40& -- &-- & 78.7k & 1.7k &0 \\ 
Poker Stars& 10393400 & 1750000 & 15000& 160 & 25& -- &-- & 133.0k\textsuperscript{*} & 27 &0 \\ 
WSOP & 60348000 & 8000000 & 15000& 1000 & 30& -- &-- & 462.3k\textsuperscript{*} & 17 &0 \\ 
\hhline{|b:======:b:==:b:===:b|}
\multicolumn{11}{l}{\TstrutBig\specialcellleft{\textsuperscript{$\dagger$}\footnotesize{Contest is identical to the contest in the preceding row, but with the prize pool rounded to a nearby number in} \\ \footnotesize{an effort to force a solution involving only nice numbers to exist.}}}\\
\multicolumn{11}{l}{\Tstrut\textsuperscript{*}\footnotesize{Heuristic produced solution with nice number constraint violation for a single bucket.}}
\end{tabular}}
\vspace{-1em}
\end{table*}

As expected, when it succeeds in finding a solution, the integer program consistently outperforms the heuristic. However, the difference is rarely large, with heuristic cost typically below 5x that of the IP. Furthermore, the heuristic runs in less than 1.5 seconds for even the most challenging contests. Its ability to generate payouts when no solution to Problem 4.1 exists is also a substantial advantage: it always returns a solution, but potentially with a minor constraint violation.

\smallskip
\noindent  \textbf{Constraint Violations.}
In our experiments, any such heuristic violation was a Nice Number violation, with no Bucket Size Monotonicity violations observed. 6 of the 7 Nice Number violations are unavoidable given the input minimum prize and prize pool. For example, the Fishing League Worldwide (FLW) fishing tournament has a prize pool of \$751,588 and a minimum prize of \$9,000. Since all nice numbers greater than or equal to \$9000 are multiples of \$5000, it is impossible to construct a fully satisfiable set of payouts summing to \$751,588. In all cases besides one (the PGA tournament with prize pool \$9,715,981) simply rounding the prize pool to a nearby number produced an input for which the heuristic output a solution with no constraint violations. However, in settings where the prize pool \emph{must} be a non-nice number (i.e., cannot be rounded up or down, for whatever reason), our heuristic's flexibility is an advantage over the more rigid integer program.

\vspace{-.5em}
\begin{figure*}[h]
\centering
\begin{subfigure}{0.33\textwidth}
\centering
\includegraphics[width=.95\textwidth]{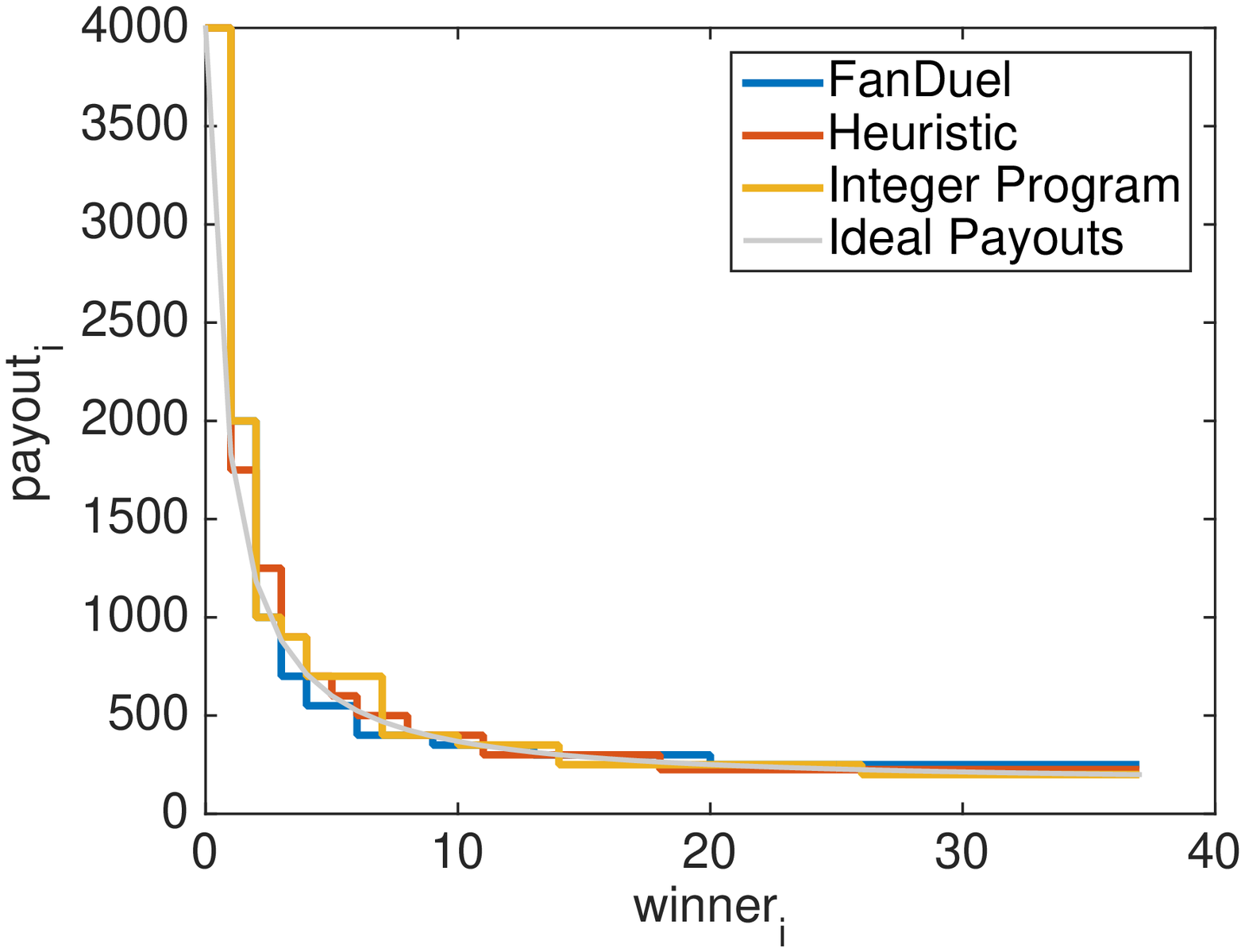} 

\caption{FanDuel, Baseball}
\end{subfigure}%
\begin{subfigure}{0.33\textwidth}
\centering
\includegraphics[width=.92\textwidth]{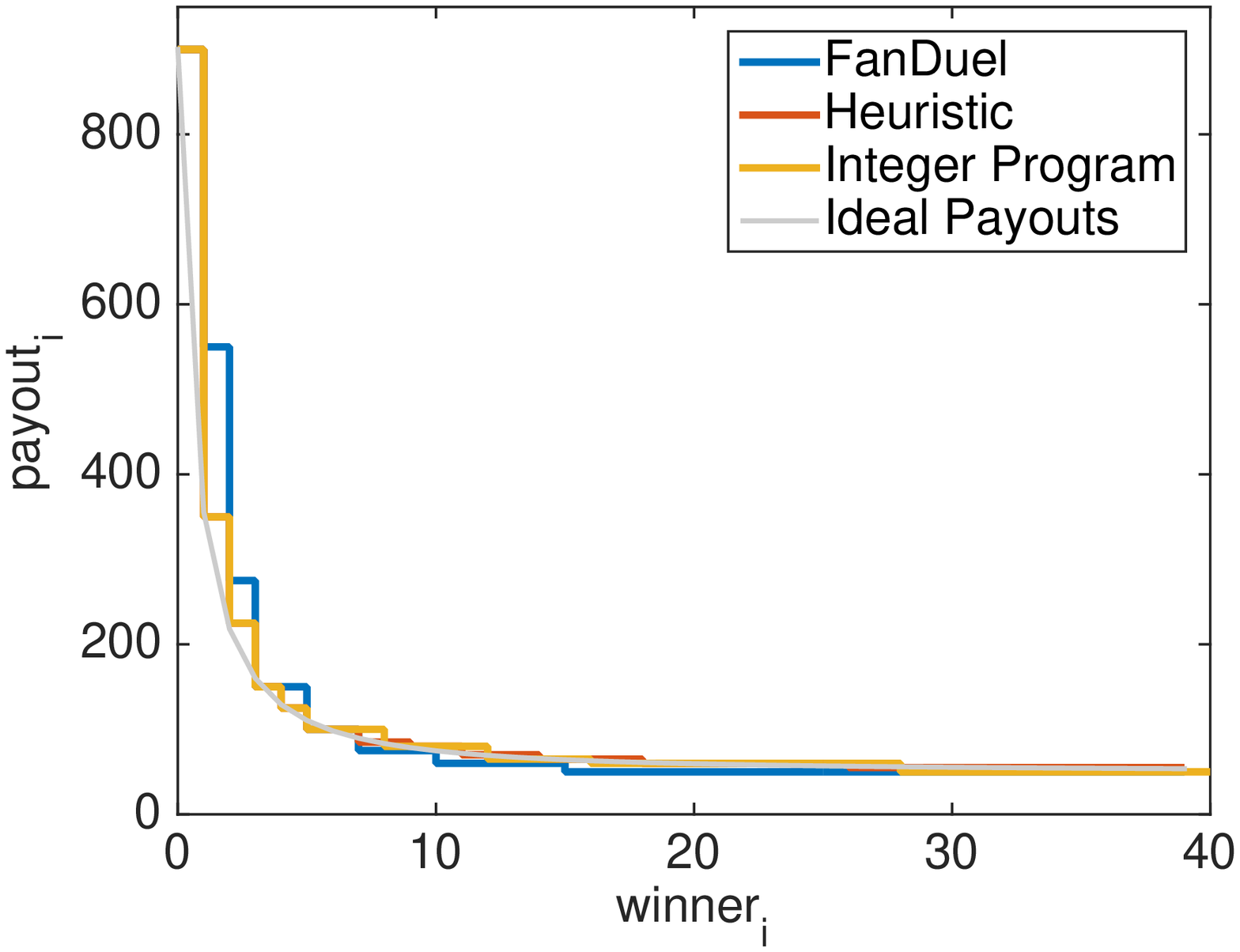}

\caption{FanDuel, Baseball}
\label{subfig:fduel2}
\end{subfigure}
\begin{subfigure}{0.33\textwidth}
\centering
\includegraphics[width=.95\textwidth]{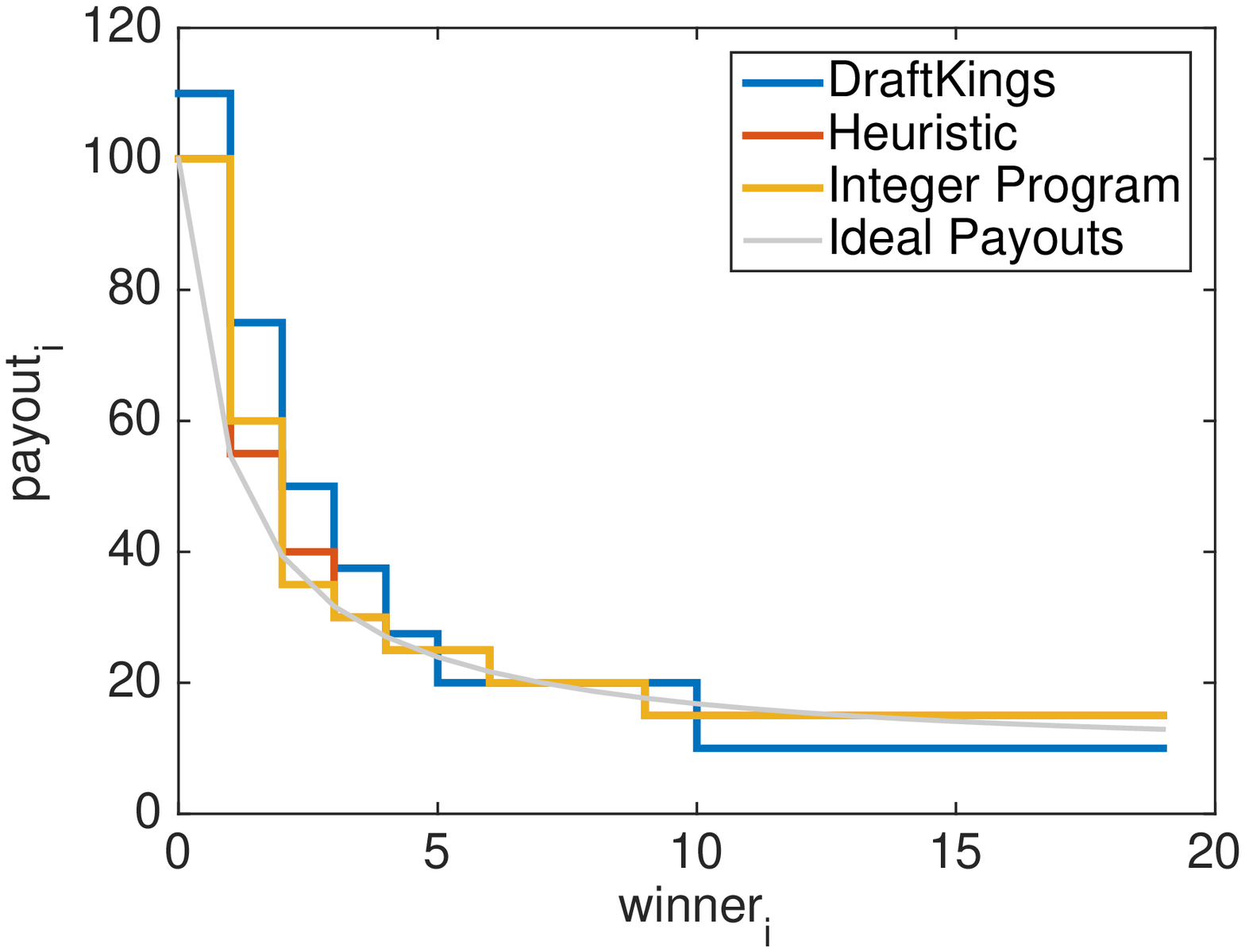}

\caption{DraftKings, Baseball}
\label{subfig:dkings1}
\end{subfigure}
\vspace{-.5em}
\caption{Payout structures for small daily fantasy contests}
\label{fig:small_contests}
\vspace{-1.25em}
\end{figure*}

\medskip
\noindent \textbf{Qualitative Comparison.}
Beyond quantitative measures to compare algorithms for Problem \ref{main_problem}, evaluating our two stage framework as a whole requires a more qualitative approach. Accordingly, we include plots comparing our generated payout structures to existing published structures. 

For small fantasy sports contests (Figure \ref{fig:small_contests}) both of our algorithms match payouts from FanDuel and DraftKings extremely well, often overlapping for large sections of the discretized payout curve. To our knowledge, FanDuel and DraftKings have not publicly discussed their methods for computing payouts; their methods may involve considerable manual intervention, so matching these structures algorithmically is encouraging.
In some cases, notably in Figures \ref{subfig:dkings1} and \ref{subfig:fduel2}, our payout curve is ``flatter'', meaning there is a smaller separation between the top prizes and lower prizes. Many poker and fantasy sports players prefer flatter structures due to reduced payout variance \cite{roto1,roto2}. 

\vspace{-1em}
\begin{figure*}[h]
\centering
\begin{subfigure}{0.48\textwidth}
\centering
\includegraphics[width=.75\textwidth]{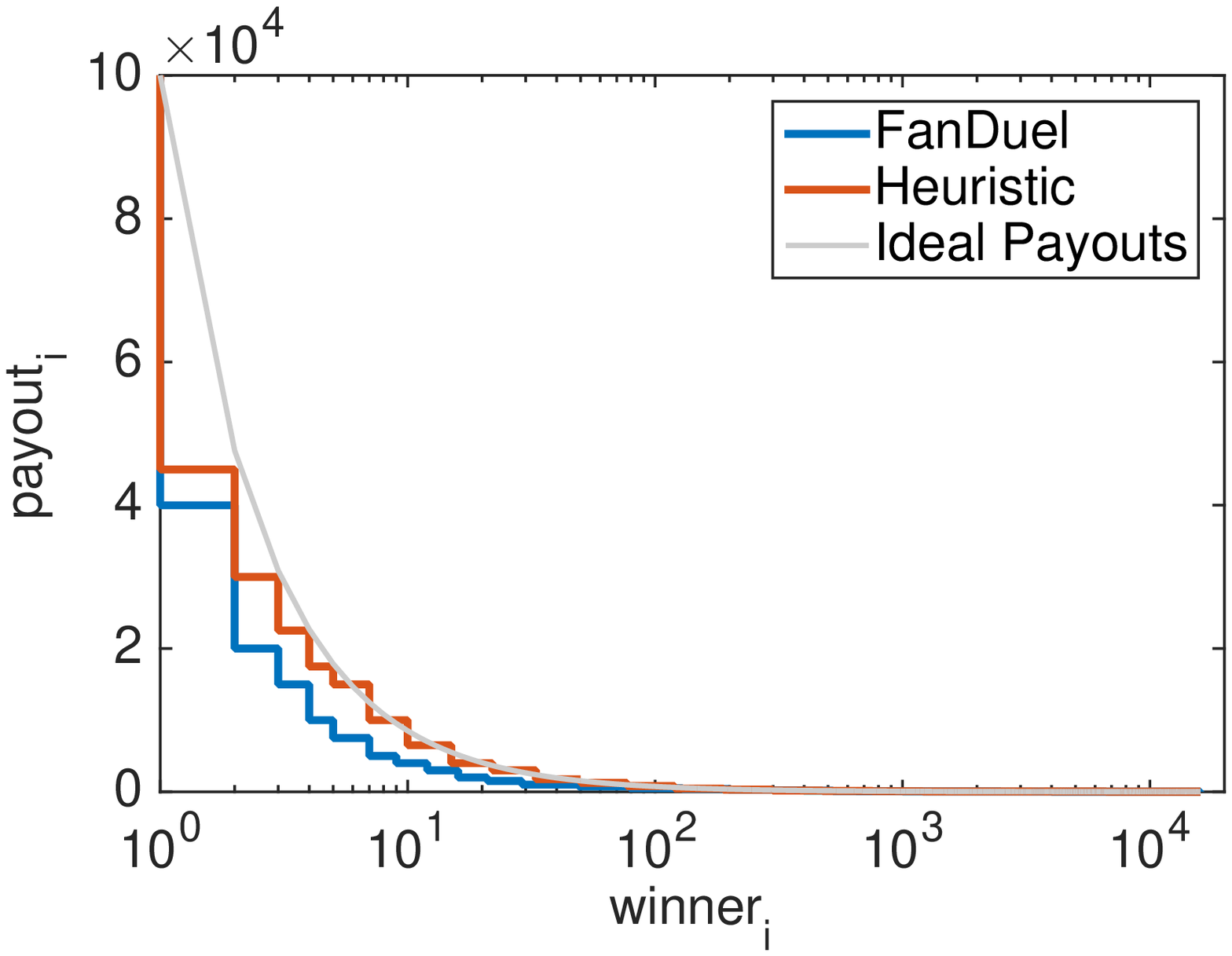} 
\vspace{-.5em}
\caption{FanDuel, Football}
\label{fig:large_contests_fd}
\end{subfigure}%
\hfill
\begin{subfigure}{0.48\textwidth}
\centering
\includegraphics[width=.75\textwidth]{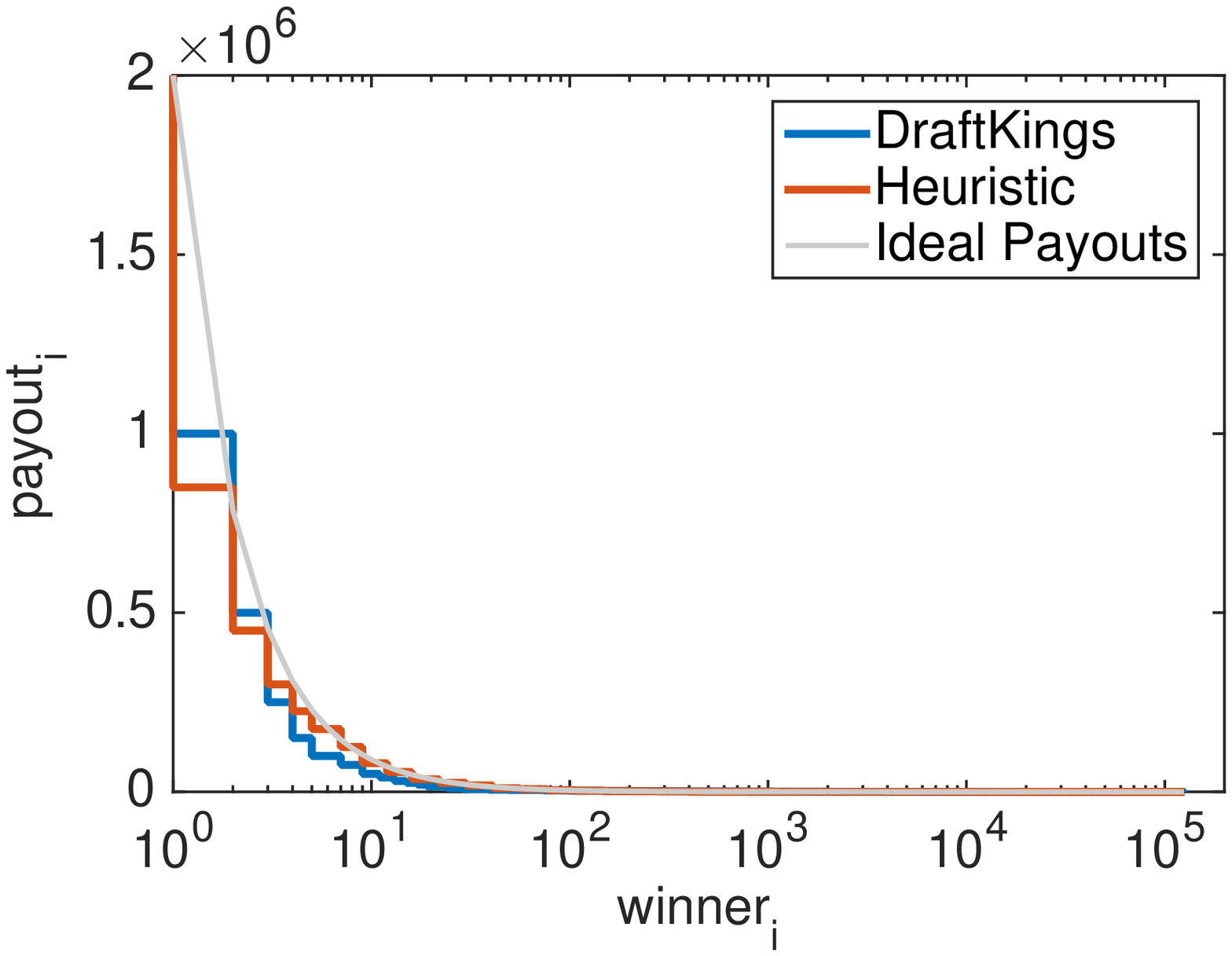}
\vspace{-.5em}
\caption{DraftKings, Football}
\label{fig:large_contests_dk}
\end{subfigure}
\vspace{-.75em}
\caption{Payout structures for large daily fantasy contests}
\label{fig:large_contests}
\vspace{-1em}
\end{figure*}
We compare larger fantasy contest payouts in Figure \ref{fig:large_contests}, plotting the $x$ axis on a log scale due to a huge number of winning places (16,000 and 125,000 for  \ref{fig:large_contests_fd} and  \ref{fig:large_contests_dk} respectively).
Again our results are very similar to those of FanDuel and DraftKings.

We also show that our algorithms can easily construct high quality payout structures for non-fantasy tournaments, avoiding the difficulties discussed in Section \ref{sec:prior_work}. Returning to the Bassmaster example from Figure \ref{fig:bmaster_failure}, we show more satisfying structures generated by our IP and heuristic algorithm in Figure \ref{fig:bmaster_fixed}. For the IP, we rounded the prize pool to \$190,000 from \$190,700 to ensure a solution to Problem \ref{integer_program} exists. However, even for the original prize pool, our heuristic solution has just one non-nice number payout of \$2,700 and no other constraint violations. Both solutions avoid the issue of non-monotonic bucket sizes exhibited by the published Bassmaster payout structure.

\vspace{-1em}
\begin{figure}[h]
\centering
\includegraphics[width=.6\textwidth]{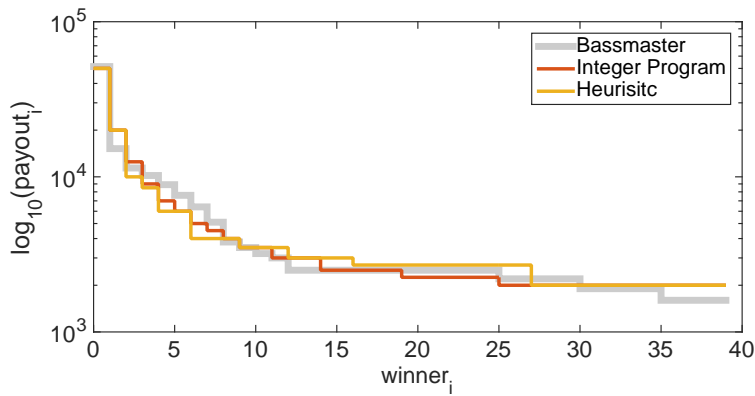}
\vspace{-.75em}
\caption{Alternative Bassmaster Open 2015 payouts}
\label{fig:bmaster_fixed}
\vspace{-1.25em}
\end{figure}


We conclude with an example from the popular World Series of Poker (WSOP). The 2015 Main Event payed out \$60,348,000 to 1000 winners in  31 buckets with far from ``nice number'' prizes \cite{wsop_events}.
However, several months prior to the tournament, organizers published a very different \emph{tentative} payout structure \cite{pokerPayoutsFailure}, that appears to have attempted to satisfy many of the constraints of Problem \ref{main_problem}: it uses mostly nice number prizes and nearly satisfies our Bucket Size Monotonicity constraint. 

This tentative structure (visualized in Figure \ref{fig:wsop_tent}) suggests that WSOP organizers originally sought a better solution than the payout structure actually used. Perhaps the effort was abandoned due to time constraints: the WSOP prize pool is not finalized until just before the event.


 We show in Table \ref{tab:wsop} that our heuristic can rapidly (in 17 milliseconds) generate an aesthetically pleasing payout structure for the final prize pool with the initially planned top prize of \$8 million, and just one minor nice number violation (our final bucket pays \$20,150). Our output and the actual WSOP payout structure are also compared visually in Figure \ref{fig:wsop_tent} to the power curve used in the first stage of our algorithm. In keeping with the WSOP tradition of paying out separate prizes to places 1-9 (the players who make it to the famous ``November Nine'' final table) we run our heuristic with 9 guaranteed singleton buckets instead of 4 (see Section \ref{sec:heuristic} for details).

\vspace{.5em}
  \begin{minipage}{\textwidth}
  \begin{minipage}[t]{0.49\textwidth}
    \centering
    \includegraphics[width=.79\textwidth]{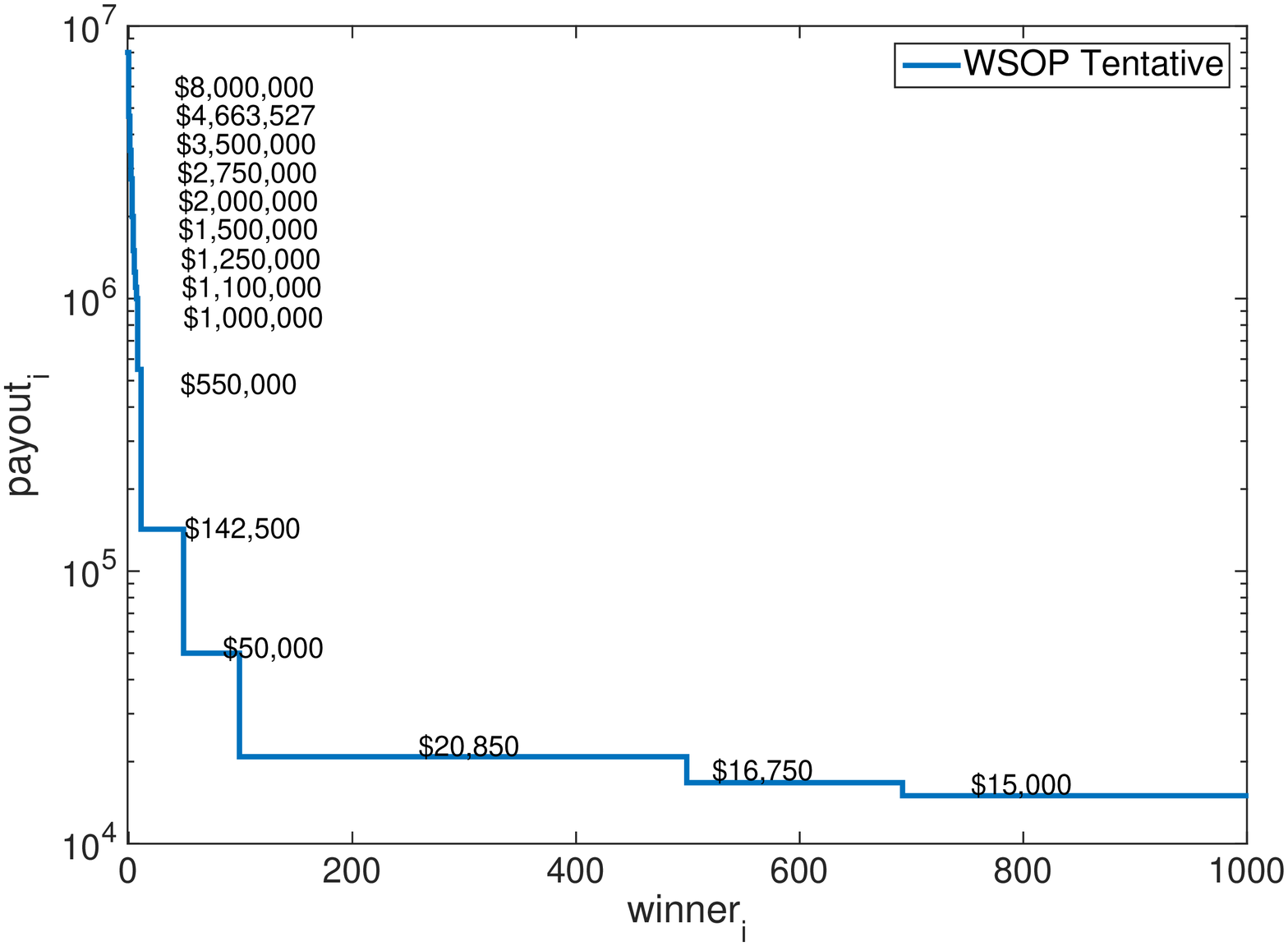} 
\vspace{-.5em}
    \includegraphics[width=.79\textwidth]{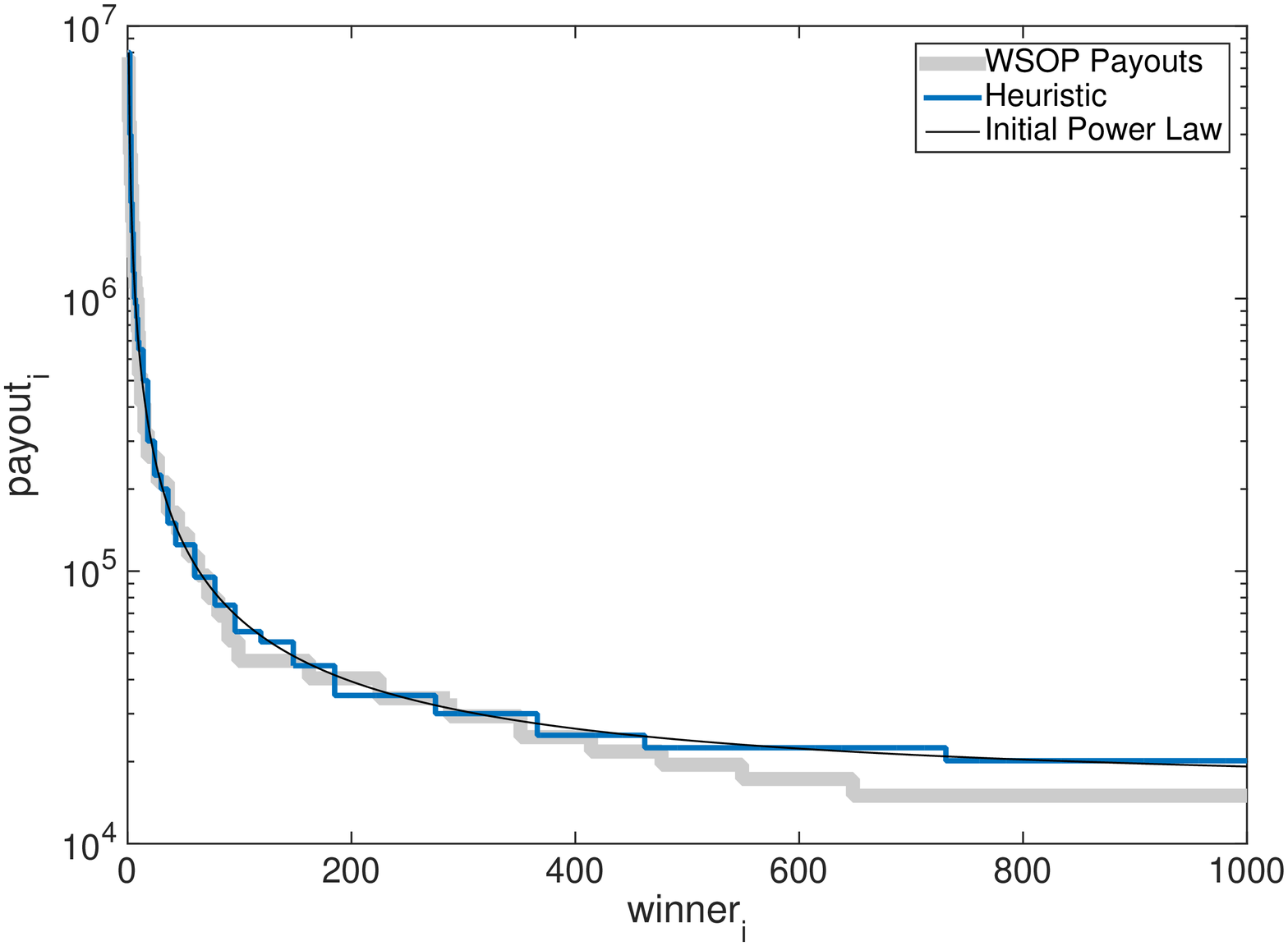}
\vspace{-.17em}
    \captionof{figure}{WSOP 2015 tentative prizes and our alternative ``nice'' payout structure.}
	\label{fig:wsop_tent}
  \end{minipage}
  \hfill
  \begin{minipage}[t]{0.49\textwidth}
    \centering
    \tiny
\begin{tabular}{||c|c||c|c||} 
\hhline{|t:==:t:==:t|}
\multicolumn{2}{||c||}{\specialcell{2015 WSOP \\Payouts}} & \multicolumn{2}{|c||}{\specialcell{Our Alternative \\Payouts}} \\
\hline
Place & Prize & Place & Prize \\
\hline
1 & \$7,680,021 & 1& \$8,000,000 \\
2 & \$4,469,171 & 2& \$4,000,000 \\
3 & \$3,397,103 & 3& \$2,250,000 \\
4 & \$2,614,558  & 4& \$1,750,000 \\
5 & \$1,910,971 & 5& \$1,250,000 \\
6 & \$1,426,072 & 6& \$1,000,000\\
7 & \$1,203,193 & 7 & \$950,000 \\
8 & \$1,097,009 & 8& \$850,000 \\
9 & \$1,001,020 & 9& \$700,000\\
10 & \$756,897 & \multirow{2}{*}{10 - 13}& \multirow{2}{*}{\$650,000} \\
11 - 12& \$526,778 & & \\
13 - 15& \$411,453 & \multirow{2}{*}{14 - 17}& \multirow{2}{*}{\$500,000} \\
16 - 18& \$325,034 & & \\
\multirow{2}{*}{19 - 27}& \multirow{2}{*}{\$262,574} & 18 - 23& \$300,000 \\
& & 24 - 29& \$225,000  \\
28 - 36& \$211,821 & 30 - 35& \$200,000  \\
36 - 45& \$164,086 &  36 - 42& \$150,000 \\
46 - 54& \$137,300 & 43 - 59& \$125,000\\
55 - 63& \$113,764  & \multirow{3}{*}{60 - 77}& \$95,000 \\
64 - 72& \$96,445& &  \\
73 - 81& \$79,668 & & \\
82 - 90& \$68,624 & \multirow{2}{*}{78 - 99}& \multirow{2}{*}{\$75,000} \\
91 - 99& \$55,649& & \\
\multirow{2}{*}{100 - 162}& \multirow{2}{*}{\$46,890} & 100 - 128& \$60,000 \\
 & & 128 - 164& \$55,000 \\
163 - 225& \$40,433 & 165 - 254& \$45,000 \\
226 - 288& \$34,157 & \multirow{2}{*}{255 - 345}& \multirow{2}{*}{\$35,000} \\
289 - 351& \$29,329 & &  \\
352 - 414& \$24,622 & \multirow{2}{*}{346 - 441}& \multirow{2}{*}{\$25,000} \\
415 - 477& \$21,786 & &\\
478 - 549&  \$19,500 &\multirow{2}{*}{442 - 710} & \multirow{2}{*}{\$22,500} \\
550 - 648& \$17,282 & & \\
649 - 1000& \$15,000 & 711 - 1000& \$20,150 \\
\hhline{|b:==:b:==:b|}
\end{tabular}
      \captionof{table}{Our alternative payouts vs. actual  WSOP payouts.\label{tab:wsop}}
    \end{minipage}
  \end{minipage}

\clearpage
\bibliography{ALENEX17_PayoutsPaper}{}
\bibliographystyle{plain}


\clearpage
\appendix
\section{Dynamic Programming Details}
\label{app:dyno_proof}
\begin{reptheorem}{thm:dynamic}[Dynamic Programming Solution]
Assuming that there are $O(\log B)$ nice numbers in the range $[E,B]$, then Problem \ref{main_problem} can be solved in pseudo-polynomial time $O(rN^3B\log^2B)$.
\end{reptheorem}
\begin{proof}
We can construct a dynamic program that chooses buckets and prizes in order from the lowest awarded places to the winner: $(S_r, \Pi_r), (S_{r-1}, \Pi_{r-1}),\ldots (S_{1}, \Pi_{1})$.
Suppose you have already selected $(S_r, \Pi_r), \ldots, (S_{i+1}, \Pi_{i+1})$. Since our objective function can be separated as
$\sum_{j=1}^r \sum_{i\in S_j}(\pi_i-\Pi_{j})^2 = \sum_{j=1}^i \sum_{i\in S_j}(\pi_i-\Pi_{j})^2 +  \sum_{j=i+1}^r \sum_{i\in S_j}(\pi_i-\Pi_{j})^2$, 
the task of choosing the optimal $(S_{i}, \Pi_{i}),\allowbreak \ldots, (S_{1}, \Pi_{1})$ then reduces to solving a smaller instance of Problem \ref{main_problem}, with:
\begin{enumerate}
\item $N$ = the number of winners remaining, i.e. those not placed in $(S_r, \Pi_r), \ldots, (S_{i+1}, \Pi_{i+1})$.
\item $B$ = the prize budget remaining, i.e. whatever was not consumed by $(S_r, \Pi_r), \ldots, (S_{i+1}, \Pi_{i+1})$.
\item $E$ = the smallest nice number greater than maximum prize given so far, $\Pi_{i+1}$.
\item An additional constraint on maximum bucket size, which must be less than $|S_{i+1}|$.
\end{enumerate}
Thus, for a given tuple of (winners consumed, budget consumed, $\Pi_{i+1}$, $|S_{i+1}|$), our choice for $(S_{i}, \Pi_{i}), \ldots, (S_{1}, \Pi_{1})$ can ignore any other aspects of $(S_r, \Pi_r), \ldots, (S_{i+1}, \Pi_{i+1})$. Accordingly, it suffices to store the optimal assignment to $(S_r, \Pi_r), \ldots,\allowbreak (S_{i+1}, \Pi_{i+1})$ for each such 4-tuple. There are at worst $N\times B\times O(\log B) \times N = O(N^2 B \log B)$ tuples. 

Now, given all assignments for $(S_r, \Pi_r), \ldots, (S_{i+1}, \Pi_{i+1})$ that are optimal for some tuple, one simple way to generate the set of optimal assignments for $(S_r, \Pi_r), \ldots, (S_{i}, \Pi_{i})$ is to try adding every possible $(S_{i}, \Pi_{i})$ (there are $\leq N\log B$ of them) to each optimal assignment. For each new assignment generated, we check if it has lower cost than the current best assignment to $(S_r, \Pi_r), \ldots, (S_{i}, \Pi_{i})$ for the tuple induced by the choice of $(S_{i}, \Pi_{i})$ (or just store it if no other assignment has induced the tuple before). In total this update procedure takes time $O(N\log B) \times O(N^2 B \log B)$ time. It must be repeated $r$ times, once for each bucket added. Our total runtime is thus $O(rN^3B\log^2 B)$.
\end{proof}
The dynamic program solution requires $O(r N^2 B \log B)$ space to store the optimal assignment for each tuple at the current iteration.

\section{Integer Programming Details}
\label{app:ip_proof}
\begin{theorem}{thm:integer}[Integer Programming Solution]
Any solution satisfying the constraints of integer programming Problem \ref{integer_program} gives a solution to Problem \ref{main_problem}.
\end{theorem}

\begin{proof}
Given a solution to Problem \ref{integer_program}, a solution to Problem \ref{main_problem} is obtained by first checking each auxiliary variable. If $\tilde{x}_{j,k} = 1$, then $\Pi_j$ is set to $p_k$. The second consistency constraint ensures that for a given $j$, $\tilde{x}_{j,k}$ can equal 1 for at most one value of $k$. Thus, each $\Pi_j$ is assigned at most once. If no $\tilde{x}_{j,k}$ is equal to 1 for a given $j$, then $\Pi_j$ and accordingly $S_j$, remain unused in our solution.

The winners for bucket $S_j$ are simply all values of $i$ for which $x_{i,j,k}$ is equal to 1. Our first consistency constraint ensures that for a given $i$, $x_{i,j,k}$ is 1 exactly once, meaning that each winner is assigned to a single bucket. The third consistency constraint ensures that $x_{i,j,k} = 1$ only when $\tilde{x}_{j,k} = 1$, which ensures that each winner in bucket $S_j$ is assigned the correct payout, $\Pi_j = p_k$. 

Accordingly, any solution to Problem \ref{integer_program} produces a valid payout structure. We just need to prove that it conforms to the constraints of our original optimization problem.

 Recalling that, for a given $i$, $x_{i,j,k} = 1$ for exactly one pair $(j,k)$, it is straight forward to see that the objective function and prize pool constraint are equivalent those from Problem \ref{main_problem} as long. The bucket size monotonicity constraint is also simple to verify: $\sum_{i\in[N], k\in[m]} x_{i,j,k}$ is exactly equal to the number of winners assigned to bucket $j$, $|S_j|$. So our constraint correctly verifies that $|S_j| \leq |S_{j+1}|$ for all $j$.

For the monotonicity requirement, notice that $\sum_{k\in[m]} (k+\frac{1}{2})\cdot\tilde{x}_{j,k}$ simply equals $k + \frac{1}{2}$ whenever $\Pi_j = p_k$ in our solution. Similarly $\sum_{k\in[m]} k\cdot\tilde{x}_{j+1,k}$ equals $k$ when $\Pi_{j+1} = p_k$. Since lower values of $k$ correspond to higher prize payouts (recall that $p_1 > p_2 > \ldots p_k$), the first constraint in our integer program therefore ensures that $\Pi_j$ is strictly greater than $\Pi_{j+1}$. There are a number of ways to enforce the ``strictly'' requirement besides adding the $\frac{1}{2}$ where we did. However, this particular approach gracefully allows the constraint to be satisfied when \emph{no prize} is assigned to bucket $S_j$ (i.e. the bucket is not used in our solution. In that case $\sum_{k\in[m]} (k+\frac{1}{2})\cdot\tilde{x}_{j,k}$ will equal zero since $\tilde{x}_{j,k}$ will equal 0 for all $k$. We would run into an issue if say, bucket $S_j$ is not used but bucket $S_{j-1}$ is. Accordingly, any solution to Problem \ref{integer_program} not using $r$ buckets must leave the lowest number sets empty, in accordance with the convention from Problem \ref{main_problem}.

Finally, note that we did not include any inequalities to ensure that each bucket contains a contiguous range of winners -- for examples winners 6 through 11 or 54 through 75. We can avoid this costly constraint because it is implicitly enforced by our choice of objective function and the monotonicity of ideal payouts. Suppose some bucket contains a non-contiguous set of winners. Then for at least two winning positions $i > j$ it must be that the prize given to winner $i$, $p_{low}$, is less than the prize given to winner $j$, $p_{high}$. At the same time, we know that for our ideal payouts, $\pi_i \geq \pi_j$. We argue that such a configuration cannot exist in an optimal solution to Problem \ref{integer_program} because switching the prizes so that $i$ receives $p_{high}$ and $j$ receives $p_{low}$ gives a strictly better solution.

First off, clearly switching the prizes will not effect our other constraints: total payout amount and bucket sizes will remain unchanged. Now, our cost function is simply additive over each winner, and the cost incurred by players $i$ and $j$ in the non-contiguous ordering is:
\begin{align*}
(\pi_i - p_{low})^2 + &(\pi_j - p_{high})^2 = \\ &\pi_i^2 - 2\pi_ip_{low} + p_{low}^2 + \pi_j^2 - 2\pi_jp_{high} + p_{high}^2.
\end{align*}
When the prizes are switched, the cost is:
\begin{align*}
(\pi_i - p_{high})^2 + &(\pi_j - p_{low})^2 = \\&\pi_i^2 - 2\pi_ip_{high} + p_{high}^2 + \pi_j^2 - 2\pi_jp_{low} + p_{low}^2.
\end{align*}
The difference in cost after the switch is therefore equal to:
\begin{align*}
- 2\pi_ip_{high} - &2\pi_jp_{low} - (- 2\pi_ip_{low} - 2\pi_jp_{high}) \\
&= -2\left(\pi_ip_{high} -\pi_ip_{low} + \pi_jp_{low} -\pi_jp_{high}\right)\\
&= -2(\pi_i - \pi_j)(p_{high} -p_{low}) \leq 0.
\end{align*}
The last step follows from the fact that $(\pi_i - \pi_j) \geq 0$ and $(p_{high} -p_{low}) > 0$. Since switching payouts reduces our cost, it follows that any solution without contiguous bucket members cannot be optimal. Note that, even with monotonic ideal payouts, this would not have been the case if we had used, for example, $\ell_1$ cost $\sum_{i\in[N],j\in[r],k\in[m]} x_{i,j,k}\cdot |\pi_i - p_k|$.

With this last implicit constraint verified, we conclude the proof that Problem \ref{integer_program} can be used to obtain a valid payout structure obeying all of the constraints of our original Problem \ref{main_problem}.
\end{proof}
\end{document}